\documentclass [ 11pt]{article}

\usepackage{amsmath,amsthm,amscd,amssymb,eurosym}

\usepackage[small, centerlast, it]{caption}
\usepackage{epsfig}
\usepackage[titles]{tocloft}
\usepackage[latin1]{inputenc}
\usepackage{verbatim} 
\usepackage[active]{srcltx} 
\usepackage{fancyhdr}
\usepackage{caption}

\def\b1{{1\!\!1}}

\def\cF{{\ca F}}

\def\cH{{\ca H}}
\def\cI{{\ca I}}
\def\cL{{\mathsf L}}

\def\cS{{\ca S}}

\def\sI{\mathsf I}

\def\sH{{\mathsf H}}

\def\sP{{\mathsf P}}

\def\sS{{\mathsf S}}

\def\tr{\mbox{tr}}
\def\bC{{\mathbb C}}           

\def\bR{{\mathbb R}}

\def\mL{\mathcal L}

\def\beq{\begin{eqnarray}}
\def\eeq{\end{eqnarray}}

\newcommand{\ca}[1]{{\cal #1}}         

\def\p{\parallel}

\newcommand{\bra}[1]{\langle{#1}|}
\newcommand{\ket}[1]{|{#1}\rangle}

\newcommand{\supp}{\text{supp }}

\usepackage{sectsty}
\sectionfont{\normalsize}
\subsectionfont{\normalsize\normalfont\itshape}
\newtheoremstyle{thm}
{12pt}
{12pt}
{\itshape}
{}
{\itshape\bfseries}
{}
{1em}
{}

\theoremstyle{thm}
\newtheorem{theorem}{Theorem}

\newtheorem{proposition}[theorem]{Proposition}
\newtheorem{definition}[theorem]{Definition}

\begin{document}

\hfill{\sl  } 
\par 
\bigskip 
\par 
\rm

\begin{center}
\Large{A geometrization of quantum mutual information}

\end{center}

\vspace{0.0 cm}

\begin{center}
 Davide Pastorello\\
\small{Dept. of Mathematics, University of Trento}\\
\small{Trento Institute for Fundamental Physics and Applications (TIFPA-INFN)}\\
\small{via Sommarive 14, Povo (TN); d.pastorello@unitn.it}
\end{center}

\section*{Abstract}
It is well known that quantum mechanics admits a geometric formulation on the complex projective space as a K\"ahler manifold. In this paper we consider the notion of mutual information among continuous random variables in relation to the geometric description of a composite quantum system introducing a new measure of total correlations that can be computed in terms of Gaussian integrals.
\\
\\
\emph{Keywords}: Geometric quantum mechanics; quantum mutual information; quantum correlations.

\section{Introduction}

The initial idea to formulate quantum mechanics as a proper Hamiltonian theory in the complex projective space (where there is a natural K\"ahler structure) was proposed in \cite{Kibble} and developed later in \cite{AS, BSS, BH, DV2} for instance. Within the geometric Hamiltonian formulation the projective space $\sP(\sH)$ constructed on the Hilbert space $\sH$ of the considered quantum theory plays the role of a phase space where quantum observables are represented by scalar functions on $\sP(\sH)$ and quantum dynamics is described in terms of the flow of a Hamiltonian vector field. Moreover quantum states can be represented by probability densities to compute expectation values integrating w.r.t.\!\! a Liouville volume form like in classical statistical mechanics. In \cite{DV2} a complete characterization of the functions on $\sP(\sH)$ that describe quantum observables is provided and the observable $C^*$-algebra (in particular the quantum $\star$-product) is constructed in terms of phase space functions regardless of the notion of linear operators in the underlying Hilbert space.
\\
Quantum mechanics predicts correlations between physical systems that cannot be explained assuming {locality} and {realism}, the \emph{quantum correlations}.
If we consider the state of a bipartite quantum system given by a density matrix $\sigma$ in  $\sH_A\otimes\sH_B$ then the total correlations (both quantum and classical correlations) of $\sigma$ are quantified by the \emph{quantum mutual information}:
\beq\label{QMI}
\sI(\sigma):=\sS(\sigma_A)+\sS(\sigma_B)-\sS(\sigma)
\eeq 
where $\sS$ is the von Neumann entropy $\sS(\sigma):=-\tr(\sigma\log_2\sigma)$ and $\sigma_{A,B}$ are the reduced states calculated as partial traces w.r.t. $\sH_{A,B}$.
 In this paper we address the problem to quantify the amount of total correlations in a quantum state within the self-consistent geometric framework of quantum mechanics, we restrict to the finite-dimensional case. 
\\
In the considered geometric formulation, where quantum states are represented by probability density functions w.r.t. an invariant measure on $\sP(\sH)$ as explained in the next section, our aim is to quantify correlations in a bipartite quantum system by means of a calculation based on the notion of mutual information of continuous random variables.\\
 Let $X$ and $Y$ be continuous random variables with \emph{joint probability density} given by $f$, a measure of their mutual dependence is given by the \emph{mutual information} defined as:
\beq\label{MI}
\cI(X,Y):=\int_{\supp\!\! f} f(x,y)\log_2\left(\frac{f(x,y)}{f_X(x)f_Y(y)}\right)\, dx\, dy
\eeq
where $f_{X,Y}$ are the marginal probability densities of $f$ and the measure $dxdy$ is the product of the reference measures on the set of values of $X$ and $Y$.  
In other words $\cI$ is the Kullback-Leibler divergence between the joint probability density and the product of its marginals, so $\cI(X,Y)=0$ if and only if $X$ and $Y$ are independent. 
\\
In the next section we briefly summarize the geometric formulation of quantum mechanics on the projective space equipped with its natural symplectic and  Riemannian structures. Then we apply the definition of mutual information (\ref{MI}) in the geometric framework investigating its physical meaning, it is not the analogue of the von Neumann mutual information (\ref{QMI}) but it turns out to be a new figure of merit to quantify correlations in quantum states.

\section{Geometrization of Quantum Mechanics on the complex projective space}

Pure states of a quantum system described in the Hilbert space $(\sH, \langle\,\, |\,\,\rangle)$ are represented by the points of the projective space $\sP(\sH):=\frac{\sH}{\sim}\setminus [0]$ where, for $\psi,\varphi\in\sH$, $\psi\sim\varphi$ if and only if $\psi=\alpha\varphi$ with $\alpha\in \bC\setminus\{0\}$. $\sP(\sH)$ is connected and Hausdorff in the quotient topology. It is well-known that the map $\sP(\sH)\ni [\psi]\mapsto \ket\psi\bra\psi\in \mathfrak P_1(\sH)$, with $\p\psi\p=1$, is a homeomorphism where $\mathfrak P_1(\sH)$ is the space of rank-1 orthogonal projectors in $\sH$ equipped with the topology induced by the standard operator norm.  
\\
Let us restrict to the case $\dim\sH=n<+\infty$. The projective space $\sP(\sH)$ has a structure of a $(2n-2)$-dimensional smooth real manifold and the tangent vectors $v\in T_p\sP(\sH)$ have the form $v=-i[A_v, p]$ for some self-adjoint operator $A_v$ on $\sH$ \cite{DV2}. As a real manifold $\sP(\sH)$ can be equipped with a symplectic structure given by the following form:
\beq\label{omega}
\omega_p(u,v):=-i\kappa \,\tr([A_u, A_v]p)\qquad \kappa >0.
\eeq
The value of the constant $\kappa$ is a natural geometric degree of freedom and it can be fixed for convenience of calculus. In this paper we take the choice: $\kappa=n+1$. 
\\
$\sP(\sH)$ can be also equipped with the Riemannian structure induced by the well-known Fubini-Study metric $g$:
\beq\label{g}
g_p(u,v):=-\kappa\,\tr(p([A_u,p][A_v,p]+[A_v,p][A_u,p])).
\eeq 
One can prove that the metric $g$ is compatible with the symplectic form $\omega$ by means of the complex form $j_p:T_p\sP(\sH)\ni v\mapsto i[v,p]\in T_p\sP(\sH)$, i.e. $\sP(\sH)$ is a K\"ahler manifold.\\
As proved in \cite{DV2}, the unique regular Borel measure $\nu$ that is left-invariant w.r.t. the smooth action\footnote{The group $U(n)$ is represented on $\sP(\sH)$ by: 
$U(n)\times\sP(\sH)\ni (U,p)\mapsto UpU^{-1}\in\sP(\sH).$}
  of the unitary group $U(n)$ on $\sP(\sH)$, with $\nu(\sP(\sH))=1$, coincides to the Riemannian measure induced by the metric $g$ and to the Liouville volume form $\omega \wedge \cdots\mbox{($n-1$) times}\cdots\wedge \omega$ up to its normalization. Such a unique $U(n)$-invariant measure presents a useful characterization in terms of the standard Gaussian measure on $\sH$ (as a $2n$-dimensional real vector space)
in the following sense \cite{gibbons, DV1}: For any bounded Borel function $f:\sP(\sH)\rightarrow \bC$ we have:
\beq\label{Gauss}
\int_{\sP(\sH)} f(p) \,d\nu(p) =\frac{1}{(2\pi)^n}\int_\sH f\circ\pi\,(x)\,{e^\frac{-\parallel x\parallel^2}{2}} \,dx
\eeq
where $\pi:\sH\setminus \{0\}\rightarrow \sP(\sH)$ is the canonical projection and $dx$ is the Lebesgue measure on $\bR^{2n}$. \\
In order to give the interpretation of $\sP(\sH)$ as a quantum phase space, let us consider: The set of quantum states $\cS(\sH):=\{\sigma\in\mathcal L(\sH): \sigma\geq 0 , \tr(\sigma)=1\}$ where $\mathcal L(\sH)$ is the space of linear operators in $\sH$, the space of self-adjoint operators $\mathfrak H(\sH)$ representing the quantum observables and the set $\cF$ of bounded Borel functions on $\sP(\sH)$. As proved in \cite{DV2} assuming $\dim\sH>2$ and $\kappa=n+1$, the maps:
\beq\label{3}
\cS(\sH)\ni\sigma\mapsto\rho_\sigma\in\cF\qquad  \rho_\sigma(p):= \tr(\sigma p)\qquad\qquad\qquad\,\,
\eeq
\beq\label{4}
 \mathfrak H(\sH)\ni A\mapsto f_A\in\cF\qquad f_A(p):=(n+1)\tr(Ap)-\tr(A)
\eeq
represent the unique prescription to associate quantum states to probability densities and quantum observables to scalar functions on $\sP(\sH)$ such that quantum expectation values can be calculated as in classical mechanics:
\beq
\int_{\sP(\sH)} f_A\rho_\sigma d\mu=\tr(A\sigma),
\eeq
$ \forall \sigma\in\cS(\sH)$ and $\forall A\in\mathfrak H(\sH)$, where $d\mu=n\,d\nu$. Furthermore (\ref{4}) is the unique way to describe quantum observables as scalar functions on $\sP(\sH)$ in order to represent the solutions of Schr\"odinger equation 
\beq
i\hbar\dot p=[H,p]\qquad H\in\mathfrak H(\sH)
\eeq
 as the flow lines of the Hamiltonian vector field on $\sP(\sH)$ defined by the Hamiltonian function $f_H$ within the symplectic structure of $\sP(\sH)$ \cite{AS, DV2}.
\\
The key result to translate a finite-dimensional quantum theory from the standard linear formulation to the geometric formulation on $\sP(\sH)$ is the bijective correspondence between linear operators on $\sH$ and a class of so-called \emph{frame functions} on the projective space.
\begin{definition}
Let $\sP(\sH)$ be the projective space of the $n$-dimensional Hilbert space $\sH$ and $d_2$ be the geodesic distance induced by the Fubini-Study metric. The set $\{p_i\}_{i=1,...,n} \subset\sP(\sH)$ is called \textbf{frame} in $\sP(\sH)$ if $d_2(p_i,p_j)=\frac{\pi}{2}$ for $i\not = j$.
\\
A map $f:\sP(\sH)\rightarrow\bC$ is called \textbf{frame function} if there exists $W_f\in\bC$ such that:
\beq
\sum_{i=1}^n f(p_i)=W_f,
\eeq
for every frame $\{p_i\}_{i=1,...,n}$ of $\sP(\sH)$.
\end{definition}

\noindent
Note that the definition of \emph{frame} is nothing but a way to represent an orthonormal basis of $\sH$ onto the projectice space. 
The following theorem is proved in \cite{DV1} as a tool for an alternative proof of Gleason's theorem and applied in \cite{DV2} to the geometrization of quantum mechanics.

\begin{theorem}\label{teo1}
Let $\sH$ be a $n$-dimensional Hilbert space with $2<n<\infty$ and $\mL^2(\sP(\sH),\mu)$ be the set of square-integrable functions on $\sP(\sH)$ w.r.t. the measure $\mu$. For every frame function $f\in\mL^2(\sP(\sH),\mu)$ there exists a unique operator $A\in\mL(\sH)$ such that $f(p)=\emph{tr}(Ap)$, $\forall p\in\sP(\sH)$.
\end{theorem}

\noindent
The converse is true: Any function $\sP(\sH)\ni p\mapsto \tr(Ap)$, with $A\in\mL(\sH)$, is a frame function in $\mL^2(\sP(\sH),\mu)$, so we have a bijective correspondence to faithfully represent linear operators as functions on $\sP(\sH)$. The set $\cF^2(\sH)$ of frame functions in $\mL^2(\sP(\sH),\mu)$ can be endowed with a structure of $C^*$-algebra\footnote{The $C^*$-norm and the quantum $\star$-product in $\cF^2(\sH)$ are explicitely constructed in \cite{DV2}.} in order to obtain the observable algebra of a quantum system in terms of phase space functions. Within this picture quantum observables are the real functions in $\cF^2(\sH)$ and quantum states are particular probability densities on $\sP(\sH)$. In this sense we can state an equivalent formulation of quantum mechanics on a complex projective space that presents the general geometric structure of a classical theory in a symplectic manifold (the quantumness of the theory is algebraically encoded in the non-commutative product on $\cF^2(\sH)$).
\begin{definition}
Let $\sH$ be a Hilbert space $\sH$ with $\dim\sH=n>2$ and $\mu$ be the unique $U(n)$-invariant regular Borel measure on $\sP(\sH)$ such that $\mu(\sP(\sH))=n$.\\ A frame function $\rho:\sP(\sH)\rightarrow [0,1]$ with $\int_{\sP(\sH)}\rho\, d\mu=1$ is called \textbf{Liouville density} on $\sP(\sH)$.
\end{definition}

\noindent
Let us denote the set of Liouville densities on $\sP(\sH)$ as $\cL(\sH)$. By theorem \ref{teo1} any Liouville density $\rho\in\cL(\sH)$ describes a unique density matrix $\sigma\in\cS(\sH)$ in the following sense $\rho:p\mapsto\tr(\sigma p)$.\\
In the next section we investigate how the notion of quantum mutual information can be introduced in terms of Liouville densities on the projective space as a classical-like mutual information among continuous random variables.

\section{Quantum mutual information within geometric formulation}

The states of a bipartite quantum system $A+B$ must be described by Liouville densities on the projective space $\sP(\sH_A\otimes\sH_B)$ instead of $\sP(\sH_A)\times\sP(\sH_B)$ as suggested by a rough analogy to the phase space of classical systems. However $\sP(\sH_A)\times\sP(\sH_B)$ can be embedded in $\sP(\sH_A\otimes\sH_B)$ by the celebrated Segre embedding:
\beq
Seg:\,\sP(\sH_A)\times\sP(\sH_B) \longrightarrow \sP(\sH_A\otimes\sH_B)
\eeq
$$
([\alpha_1:\cdots:\alpha_n], [\beta_1:\cdots :\beta_m])\mapsto [\alpha_1\beta_1:\alpha_1\beta_2:\cdots: \alpha_n\beta_m]
$$
where the action of $Seg$ is expressed in homogenous coordinates  of the rays ($\dim\sH_A=n$ and $\dim\sH_B=m$). Equivalentely, one can express the action of the Segre embedding in terms of rank-1 orthogonal projectors:
\beq
Seg\,(\ket{\psi_A}\bra{\psi_A}, \ket{\psi_B}\bra{\psi_B})=\ket{\psi_A\otimes \psi_B}\bra{\psi_A\otimes \psi_B}\, , \qquad \p\psi_A\p=\p\psi_B\p=1.
\eeq
The image of $Seg$ is the well-known \emph{Segre variety} and gives the set of separable pure states. 
\\
The proposition below gives a characterization of entangled states in terms of Liouville densities on the projective space.

\begin{proposition}\label{propp}
Let $\sH_A$ and $\sH_B$ be Hilbert spaces with dimension larger than 2. The Liouville density $\rho\in\cL(\sH_A\otimes\sH_B)$ describes a separable (non-entangled) state if and only if it satisfies
\beq\label{10}
\rho\circ Seg(p_A, p_B)=\sum_n \lambda_n\,\rho_{A n}(p_A)\rho_{B n}(p_B)\qquad \forall (p_{A}, p_B)\in\sP(\sH_{A})\times\sP(\sH_B),
\eeq
 where $\{\lambda_n\}_{n}$ are statistical weights, i.e. $\lambda_n\geq 0$ with $\sum_n\lambda_n =1$, and $\{\rho_{An}\}_n\subset\cL(\sH_A)$, $\{\rho_{Bn}\}_n\subset\cL(\sH_B)$.

\end{proposition}

\begin{proof}
If $\sigma$ is a separable state then (\ref{10}) is obviously true for $\rho:p\mapsto\tr(\sigma p)$, let us prove the non-trivial implication. Consider the vector spaces $\cF^2(\sH_A)$ and $\cF^2(\sH_B)$ of square $\mu$-integrable frame functions on $\sP(\sH_A)$ and $\sP(\sH_B)$ respectively. We need to show that for any $f\in\cF^2(\sH_A)\otimes\cF^2(\sH_B)$ there exists a unique $g\in\cF^2(\sH_A\otimes\sH_B)$ such that $f=g\circ Seg$. Let $\{e_{k}\}_k$ and $\{h_{l}\}_l$ be bases of $\cF^2(\sH_A)$ and $\cF^2(\sH_B)$ respectively and consider a function $f\in\cF^2(\sH_A)\otimes\cF^2(\sH_B)$ written in terms of the basis $\{e_{k}\otimes h_{l}\}$:
\beq
f(p_A, p_B)=\sum_{kl} c_{kl}e_{k}(p_A)h_{l}(p_B).
\eeq
By theorem \ref{teo1} we have: $f(p_A, p_B)=\sum_{kl} \tr(A_k p_A)\tr(B_k p_B)$ with $A_k\in\mL(\sH_A)$ and $B_l\in\mL(\sH_B)$ that are univocally fixed for any $k$ and $l$. So we have that $f=g\circ Seg$, with $g\in\cF^2(\sH_A\otimes\sH_B)$, if and only if: 
\beq
g(p)=\sum_{kl} c_{kl} \,\tr[(A_k \otimes B_l) p].
\eeq
Let $\eta:\sP(\sH_A)\times\sP(\sH_B)\rightarrow [0,1]$ defined by $\eta(p_A, p_B):=\sum_n\lambda_n\rho_{An}(p_A)\rho_{Bn}(p_B)$ where $\rho_{A_n}$ and $\rho_{Bn}$ are  Liouville densities on $\sP(\sH_A)$ and $\sP(\sH_B)$ respectively. The unique function $\rho\in\cF^2(\sH_A\otimes\sH_B)$ satisfying $\eta=\rho\circ Seg$ is given by  $\rho(p)=\sum_n\lambda_n\tr[(\sigma_{An}\otimes\sigma_{Bn})p]$ where $\sigma_{An}$ and $\sigma_{Bn}$ are density matrices on $\sH_A$ and $\sH_B$. Therefore $\rho(p)=\tr(\sigma p)$ with $\sigma=\sum_n\lambda_n\sigma_{An}\otimes\sigma_{Bn}$, i.e. $\rho$ is a Liouville density describing a non-entangled state.
\end{proof}

\noindent
The Segre variety represents (by means of the embedding)  the \emph{classical-like} phase space of the composite system in the sense of  the cartesian product of the quantum phase spaces of the single subsystems. Thus the most natural interpretation of the above result is the following: If there is no entanglement then the Liouville density describing the considered quantum state presents the form of a classical bipartite state (i.e. a statistical mixture of products) when restricted to the Segre variety. 
\\
In order to introduce the notion of quantum mutual information we need to consider the marginal probability densities of $\rho\circ Seg:\sP(\sH_A)\times\sP(\sH_B)\rightarrow[0,1]$ with $\rho\in\cL(\sH_A\otimes\sH_B)$:
\beq\label{marginal}
\rho_A(p_B)=\int_{\sP(\sH_A)} \rho\circ Seg (p_A,p_B)\, d\mu_A(p_A).
\eeq
\\
Applying theorem \ref{teo1} it is easy to prove that the unique operator $\sigma_A\in\mL(\sH_B)$ such that $\rho_A:p_B\mapsto \tr(\sigma_Ap_B)$ is nothing but the partial trace of the density matrix associated to $\rho$.
In particular the statement of proposition \ref{propp} implies that a Liouville density $\rho$ describes a completely uncorrelated state $\sigma$ (i.e. it is factored as $\sigma=\sigma_A\otimes\sigma_B$)  if and only if it satisfies $\rho\circ Seg (p_A, p_B)=\rho_A(p_A)\rho_B(p_B)$, for all $p_{A,B}\in\sP(\sH_{A,B})$, where $\rho_{A}$ and $\rho_{B}$ are the marginal probability densities of $\rho\circ Seg$ on $\sP(\sH_A)$ and $\sP(\sH_B)$ respectively. In other words the restriction of $\rho$ to the Segre variety has the form of an uncorrelated classical state.
\\
In view of the latter comments, the key idea of our approach is to consider a quantum system as a continuous random variable valued in the complex projective space $\sP(\sH)$ (that has cardinality of the continuum as a connected Hausdorff manifold \cite{N}) equipped with the measure $\mu$ of the geometric formulation discussed in the previous section. Given a Liouville density $\rho\in\cL(\sH_A\otimes\sH_B)$ describing the state of the system $A+B$, we can calculate the mutual information between $A$ and $B$ considering $\rho\circ Seg$ as a joint probability density according to classical defintion (\ref{MI}):
\beq
\mathcal I(\rho):=\int_{\supp(\rho\circ Seg)} \rho\circ Seg (p_A,p_B)\,\,\log_2\left(\frac{\rho\circ Seg (p_A,p_B)}{\rho_A(p_A)\rho_B(p_B)}\right) \,d\mu_A(p_A)d\mu_B(p_B).
\eeq
$\mathcal I$ does not obviously correpsond to the Von Neumann mutual information calculated from the density matrix associated to $\rho$. However it presents the properties to be a good measure of total correlations in a quantum state.
In fact $\mathcal I(\rho)=0$ if and only if $\rho\circ Seg$ is the product of the marginal probability densities that is the case where there are no classical or quantum correlations in the considered state.  By proposition \ref{propp} if $\rho$ describes a separable mixed state then $\cI$ is a measure of the correlations due to the incoherent superposition in the considered state. On the other hand if we have an entangled pure state, so the correlations are purely quantum, $\cI$ quantifies  how much $\rho\circ Seg$ is different from the product of the marginals as a Kullback-Leibler divergence thus it gives a measure of entanglement for pure states. 
\\
Let us consider $\rho\in\cL(\sH_A\otimes\sH_B)$ describing a maximally entangled state and calculate $\cI(\rho)$. We can apply the following identity involving differential entropies:
\beq
\cI(\rho)=\cH(\rho_A)-\cH(\rho_A|\rho_B)
\eeq
 where $\cH(\rho_A)$ is the differential entropy\footnote{The differential entropy of a continuous random variable $X$ with density $f$ is defined as $\mathcal H(X):=\int f\log_2 f dx$ and it is well-known that it does not preserve the properties of the information entropy for discrete random variables. However the continuous mutual information $\cI$ still be a meaningful measure of mutual dependece for continuous random variables and the identity $\cI(X,Y)=\cH(X)-\cH(X|Y)$ still hold in the continuous case.} of $A$ 
and $\cH(\rho_A|\rho_B)$ is the differential entropy of $A$ conditioned on $B$ 
 i.e. it is the entropy calculated from the Liouville density describing the state of subsystem $A$ when the subsystem $B$ is in a known pure state. Considering the system in the maximally entangled state, if one observes the subsystem $B$ in the pure state represented by the ray $p_B\in\sP(\sH_B)$ during some local measurement process\footnote{In the sense of \emph{post-selection mapping}.}  
then the pure state of subsystem $A$ after the measurement is known deterministically, so one would suspect that the conditional entropy is zero, however $\cH(\rho_A|\rho_B)$ is non-zero. In fact the differential entropy given by a Liouville density describing a pure state is non-zero but a constant that does not depend on the pure state or Hilbert space dimension as stated by the next proposition.
\begin{proposition}
For any Liouville density $\rho_\sigma\in\cL(\sH)$ describing a pure state $\sigma$ the related differential entropy is:
\beq
\cH(\rho_\sigma)=-\int_{\sP(\sH)} \rho_\sigma(p) \log_2(\rho_\sigma(p))\, d\mu(p)=(2\gamma-2)\log_2 e-2
\eeq
where $\mu$ is the usual invariant measure on $\sP(\sH)$ and $\gamma$ is the Euler-Mascheroni constant.
\end{proposition}

\begin{proof}
The value of $\cH(\rho_\sigma)$ does not depend on $\sigma$ because of the transitive action of the unitary group on $\sP(\sH)$ and the invariance of $\mu$ under unitary transformations. Let $U$ be a unitary operator on $\sH$ and consider the pure state $\sigma'=U\sigma U^*$, we have $\rho_{\sigma'}(p)=\rho_\sigma(U^*p U)$:
\beq
\cH(\rho_{\sigma'})=-\int_{\sP(\sH)} \rho_\sigma(U^*pU) \log_2(\rho_\sigma(U^*pU))\, d\mu(p)\qquad\qquad\qquad\quad
\eeq 
$$\,=-\int_{\sP(\sH)} \rho_\sigma(U^*pU) \log_2(\rho_\sigma(U^*pU))\, d\mu(U^*pU)=\cH(\rho_\sigma).$$
Assume to fix an orthonormal basis of the $n$-dimensional Hilbert space $\sH$ and let us calculate $\cH(\rho_\sigma)$ for the pure state given by $\sigma=\ket\psi\bra\psi=$ diag$(1,0,...,0)$ in the considered coordinates. Applying (\ref{Gauss}) we have:
\beq
\cH(\rho_\sigma)=-\frac{1}{(2\pi)^n}\int_{\sH}|\langle \psi|x\rangle|^2\log_2 (|\langle \psi|x\rangle|^2) e^{-\frac{\parallel x\parallel^2}{2}} dx.
\eeq
Since $|\langle\psi|x\rangle|^2=|x_1|^2$ where $x_1=\alpha_1+i\beta_1$ is the first complex component of the vector $x$ w.r.t. the fixed basis and $\sH$ is seen as a $2n$-dimensional real space, we can re-write the integral as:
\beq
\cH(\rho_\sigma)=-\frac{1}{(2\pi)^n}\int_{\bR^{2n}}(\alpha_1^2+\beta_1^2)\log_2(\alpha_1^2+\beta_1^2)e^{-\frac{1}{2} \sum_{i=1}^n \alpha_i^2+\beta_i^2} \,\,d\alpha_1 d\beta_1 \cdots d\alpha_nd\beta_n\qquad
\eeq
$$\quad=  -\frac{1}{(2\pi)^n}\int_{\bR^{2}}(\alpha_1^2+\beta_1^2)\log_2(\alpha_1^2+\beta_1^2)e^{-\frac{1}{2} \alpha_1^2+\beta_1^2} \,\,d\alpha_1 d\beta_1 \left( \int_{-\infty}^{+\infty} e^{-\frac{1}{2}y^2} dy\right)^{2n-2} $$
$$=-\frac{1}{2\pi}\int_0^{2\pi} d\theta \int_0^{+\infty} r^3 \log_2(r^2)e^{-\frac{r^2}{2}}  dr =  (2\gamma-2)\log_2 e-2 \qquad\qquad\qquad\,\,$$
\\
where we used the known integral $\int_0^{+\infty} x^3\log_2(x^2)\exp(-x^2/2)dx=2+(2-2\gamma)\log_2 e$.
\end{proof}

\noindent
Let us give a remark about the result above: Even if a pure state is represented by a single point $p_0\in\sP(\sH)$, the differential entropy of a quantum system in a pure state is non-zero. In fact the Liouville density $\rho$ describing a pure state is not a Dirac delta centered in $p_0$, like a sharp classical state, but it is a \emph{smeard} distribution encoding the statistic produced by any possible measurement process\footnote{Note that we can describe a pure state $\sigma=\ket\psi\bra\psi$ as a Dirac mass if and only if we consider the outcomes of a fixed PVM-measurement $\{P_i\}_i$ such that $\sigma\in \{P_i\}_i$.} on the system.
\\
Returning to the calculation of $\cI(\rho)$ where $\rho\in\cL(\sH_A\otimes\sH_B)$ describes a maximally entangled state, the marginal density is the constant function $\rho_A=d^{-1}$, where $d=\dim\sH_A=\dim\sH_B$, its differential entropy is:
\beq
\cH(\rho_A)=-\int_{\sP(\sH_A)} \rho_A(p)\log_2(\rho_A(p))\, d\mu_A(p)=\log_2d,
\eeq
which corresponds to the value of the Von Neumann entropy of the associated reduced density matrix in this particular case. Thus we have:
\beq
\cI(\rho)=\log_2d+2+(2-2\gamma)\log_2e\simeq \log_2d+3.22,
\eeq
on the other hand the well-known von Neumann mutual information (\ref{QMI}) of a maximally entangled state $\sigma\in\cS(\sH_A\otimes\sH_B)$ is $\sI(\sigma)=2\log_2d$.\\
\noindent
After the discussion above, we can interpret $\cI$ as a measure to quantify the total amount of correlations inside a bipartite quantum system. Nevertheless it is not the direct translation of the von Neumann mutual information to the geometric formulation.
\\ 
 From the viewpoint of the direct calculation of $\cI$ starting from a density matrix, one can exploit the fact that the measure $\mu$ on $\sP(\sH)$ is the image of the standard Gaussian measure on decomplexified $\sH$ by means of the canonical projection. Let $\sigma\in\cS(\sH_A\otimes \sH_B)$ be the density matrix associated to the Liouville density $\rho\in\cL(\sH_A\otimes\sH_B)$, then:
\beq\label{MImatrix}
\cI(\rho)=\int_{X} \langle x\otimes y|\sigma\, x\otimes y\rangle\, \log_2\left( \frac{\langle x\otimes y|\sigma\, x\otimes y\rangle}{\langle x|\sigma_A x\rangle\langle y|\sigma_B y\rangle}  \right) d\mu_{G_A}(x)d\mu_{G_B}(y)
\eeq  
where $X:=\{(x,y)\in\sH_A\times\sH_B \,:\,  \langle x\otimes y|\sigma\, x\otimes y\rangle >0  \}$, $\mu_{G_A}$ and $\mu_{G_B}$ denote the standard Gaussian measure on decomplexified $\sH_A$ and $\sH_B$  respectively, the integral is calculated w.r.t. the product of the Gaussian measures.


\section{Conclusions}

In the present work we have shown that the notion of mutual information between continuous random variables can be used to quantify the total amount of correlations among quantum systems. The approach is formalized in the geometric framework where quantum states are described by probability density functions on the complex projective space equipped with a unitary invariant Borel measure.   
Such a classical-like mutual information is a proposed estimator which plays the role of von Neumann mutual information within the considered geometric formulation.\!
 Moreover (\ref{MImatrix}) is the formula to calculate the classical-like mutual information starting from a density matrix then $\cI$ can be also considered a figure of merit in the standard formulation of QM. Let us stress that integration w.r.t. \!$\mu$ to calculate $\cI$ or  differential entropy can be performed as a Gaussian integral on decomplexified Hilbert space, an example of calculation of this kind is given in the proof of proposition 5. Matter for future works could be the in-depth analysis of the interplay between $\cI$ and other quantities of quantum information theory.

\section*{Acknowledgements}
 
The work of the author is supported by:

\vspace{-0.4cm}

\begin{center}
 \includegraphics[width=4cm]{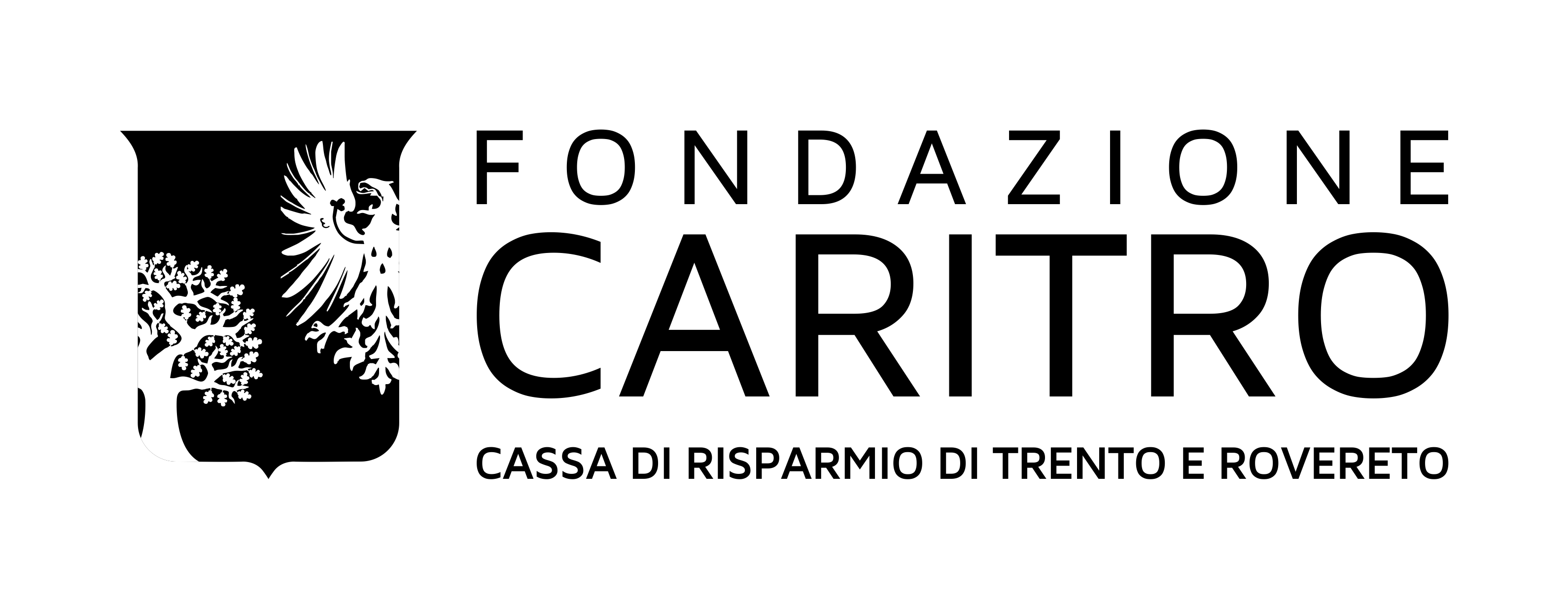}
\end{center}

\end{document}